\documentclass[runningheads,envcountsame,envcountsect,11pt]{llncs}
\usepackage[a4paper,margin=0.85in]{geometry}
\renewcommand{\orcidID}[1]{} 

\usepackage{graphicx}

\usepackage{hyperref}
\usepackage[table,xcdraw]{xcolor}

\usepackage{amsmath,amssymb}

\newcommand{\onestep}{\leftrightarrow}
\newcommand{\sevstep}{\leftrightsquigarrow}

\newcommand{\TAR}[1]{\mathsf{TAR}(#1)}
\newcommand{\TS}{\mathsf{TS}}
\newcommand{\TJ}{\mathsf{TJ}}

\newcommand{\TARrule}{\mathsf{TAR}}

\newcommand{\symdiff}[2]{#1 \vartriangle #2}

\newcommand{\localqed}{\hfill $\diamondsuit$}

\newcommand{\mw}{\ensuremath{\mathrm{mw}}}
\newcommand{\nd}{\ensuremath{\mathrm{nd}}}

\spnewtheorem{observation}[theorem]{Observation}{\bfseries}{\itshape}
\spnewtheorem{myclaim}[theorem]{Claim}{\itshape}{\upshape}

\newcommand{\figdir}{.} 
\newcommand{\figref}[1]{\figurename~\ref{#1}}

\usepackage{color}
\definecolor{lightblue}{rgb}{0.5,0.5,1.0}
\definecolor{darkred}{rgb}{0.8,0,0}
\definecolor{darkgreen}{rgb}{0,0.5,0}
\definecolor{darkblue}{rgb}{0,0,0.5}

\usepackage{algorithm}
\usepackage[noend]{algpseudocode}

\title{Independent Set Reconfiguration Parameterized by Modular-Width\thanks{%
Partially supported by JSPS and MAEDI under the Japan-France Integrated Action Program (SAKURA) Project GRAPA 38593YJ,
and by JSPS/MEXT KAKENHI Grant Numbers JP24106004, JP17H01698, JP18K11157, JP18K11168, JP18K11169, JP18H04091, JP18H06469.}}
\author{R\'{e}my Belmonte\inst{1}\orcidID{0000-0001-8043-5343} \and
Tesshu Hanaka\inst{2}\orcidID{0000-0001-6943-856X} \and
Michael Lampis\inst{3}\orcidID{0000-0002-5791-0887} \and
Hirotaka Ono\inst{4}\orcidID{0000-0003-0845-3947} \and
Yota Otachi\inst{5}\orcidID{0000-0002-0087-853X}}
\authorrunning{Belmonte et al.}
\institute{
The University of Electro-Communications, Chofu, Tokyo, 182-8585, Japan\\
\email{remy.belmonte@uec.ac.jp}
\and
Chuo University, Bunkyo-ku, Tokyo, 112-8551, Japan\\
\email{hanaka.91t@g.chuo-u.ac.jp}
\and
Universit\'{e} Paris-Dauphine, PSL University, CNRS, LAMSADE 75016, Paris, France\\
\email{michail.lampis@dauphine.fr}
\and
Nagoya University, Nagoya, 464-8601, Japan\\
\email{ono@nagoya-u.jp}
\and
Kumamoto University, Kumamoto, 860-8555, Japan\\
\email{otachi@cs.kumamoto-u.ac.jp}}
\begin{document}

\maketitle              
\begin{abstract}
\textsc{Independent Set Reconfiguration} is one of the most well-studied problems in the setting of combinatorial reconfiguration.
It is known that the problem is PSPACE-complete even for graphs of bounded bandwidth.
This fact rules out the tractability of parameterizations by most well-studied structural parameters
as most of them generalize bandwidth.
In this paper, we study the parameterization by modular-width, which is not comparable with bandwidth.
We show that the problem parameterized by modular-width is fixed-parameter tractable
under all previously studied rules $\TARrule$, $\TJ$, and $\TS$.
The result under $\TARrule$ resolves an open problem posed by Bonsma~[WG 2014, JGT 2016].

\end{abstract}


\section{Introduction}
In a reconfiguration problem,
we are given an instance of a search problem together with two feasible solutions.
The algorithmic task there is to decide whether one solution can be transformed to the other by a sequence of prescribed local modifications
while maintaining the feasibility of intermediate states.
Recently, reconfiguration versions of many search problems have been studied~(see \cite{vandenHeuvel13,Nishimura18}).

\textsc{Independent Set Reconfiguration} is one of the most well-studied reconfiguration problems.
In this problem, we are given a graph and two independent sets.
Our goal is to find a sequence of independent sets that represents a step-by-step modification
from one of the given independent sets to the other.
There are three local modification rules studied in the literature: 
Token Addition and Removal ($\TARrule$)~\cite{Bonsma16,KaminskiMM12,MouawadN0SS17},
Token Jumping ($\TJ$)~\cite{BonsmaKW14,BousquetMP17,ItoDHPSUU11,ItoKOSUY14,ItoKO14}, and
Token Sliding ($\TS$)~\cite{BelmonteKLMOS19,BonamyB17,DemaineDFHIOOUY15,Fox-EpsteinHOU15,HearnD05,HoangU16,LokshtanovM18}.
Under $\TARrule$, given a threshold $k$,
we can remove or add any vertices as long as the resultant independent set has size at least $k$.
(When we want to specify the threshold $k$, we call the rule $\TAR{k}$.)
$\TJ$ allows to swap one vertex in the current independent set with another vertex not dominated by the current independent set.
$\TS$ is a restricted version of $\TJ$ that additionally asks the swapped vertices to be adjacent.

It is known that \textsc{Independent Set Reconfiguration} is PSPACE-complete under all three rules 
for general graphs~\cite{ItoDHPSUU11},
for perfect graphs~\cite{KaminskiMM12}, and for planar graphs of maximum degree~3~\cite{HearnD05} (see \cite{BonsmaKW14}).
For claw-free graphs, the problem is solvable in polynomial time under all three rules~\cite{BonsmaKW14}.
For even-hole-free graphs (graphs without induced cycles of even length),
the problem is known to be polynomial-time solvable under $\TARrule$ and $\TJ$~\cite{KaminskiMM12},
while it is PSPACE-complete under $\TS$ even for split graphs~\cite{BelmonteKLMOS19}.
Under $\TS$, forests~\cite{DemaineDFHIOOUY15} and interval graphs~\cite{BonamyB17} form maximal known subclasses of even-hole-free graphs
for which \textsc{Independent Set Reconfiguration} is polynomial-time solvable.
For bipartite graphs,
the problem is PSPACE-complete under $\TS$ 
and, somewhat surprisingly, it is NP-complete under $\TARrule$ and $\TJ$~\cite{LokshtanovM18}.

\textsc{Independent Set Reconfiguration} is studied also in the setting of parameterized computation.
(See the recent textbook~\cite{CyganFKLMPPS15} for basic concepts in parameterized complexity.)
It is known that there is a constant $b$ such that the problem is PSPACE-complete under all three rules
even for graphs of bandwidth at most $b$~\cite{Wrochna18}.
Since bandwidth is an upper bound of well-studied structural parameters such as pathwidth, treewidth, and clique-width,
this result rules out FPT (and even XP) algorithms with these parameters.
Given this situation, Bonsma~\cite{Bonsma16} asked whether \textsc{Independent Set Reconfiguration}
parameterized by modular-width is tractable under $\TARrule$ and $\TJ$.
The main result of this paper is to answer this question by presenting
an FPT algorithm for \textsc{Independent Set Reconfiguration} under $\TARrule$ and $\TJ$ parameterized by modular-width.
We also show that under $\TS$ the problem allows a much simpler FPT algorithm.

Our results in this paper can be summarized as follows:\footnote{%
The $O^*(\cdot)$ notation suppresses factors polynomial in the input size.}
\begin{theorem}
Under all three rules $\TARrule$, $\TJ$, and $\TS$,
\textsc{Independent Set Reconfiguration} parameterized by modular-width $\mw$
can be solved in time $O^{*}(2^{\mw})$.
\end{theorem}
In Section~\ref{sec:TAR}, we give our main result for $\TARrule$ (Theorem~\ref{thm:TAR}),
which implies the result for $\TJ$ (Corollary~\ref{cor:TJ}).
The FPT algorithm under $\TS$ is given in Section~\ref{sec:TS} (Theorem~\ref{thm:TS}).


\section{Preliminaries}
\label{sec:pre}

Let $G = (V,E)$ be a graph.  For a set of vertices $S \subseteq V$, we denote
by $G[S]$ the subgraph induced by $S$. 
For a vertex set $S \subseteq V$, we denote by $G-S$ the graph $G[V \setminus S]$. 
For a vertex $u\in V$, we write $G-u$ instead of $G - \{u\}$. 
For $u,v \in S$, we denote $S \cup \{u\}$ by $S + u$
and $S \setminus \{v\}$ by $S - v$, respectively.
We use $\alpha(G)$ to denote the size of
a maximum independent set of $G$. For two sets $S,R$ we use $\symdiff{S}{R}$ to
denote their symmetric difference, that is, the set $(S\setminus
R)\cup(R\setminus S)$. For an integer $k$ we use $[k]$ to denote the set
$\{1,\ldots,k\}$.
For a vertex $v \in V$, its (\emph{open}) \emph{neighborhood} is denoted by $N(v)$.
The \emph{open neighborhood} of a set $S \subseteq V$ of vertices is defined as $N(S) = \bigcup_{v \in S} N(v) \setminus S$.
A \emph{component} of $G$ is a maximal vertex set $S \subseteq V$ such that
$G$ contains a path between each pair of vertices in $S$.

In the rest of this section, we are going to give definitions of the terms used in the following formalization of the main problem:
\begin{itemize}
  \setlength{\itemsep}{0pt}
  \item[] \textbf{Problem:} \textsc{Independent Set Reconfiguration} under $\TAR{k}$
  \item[] \textbf{Input:} A graph $G$, an integer $k$,  and independent sets $S$ and $S'$ of $G$.
  \item[] \textbf{Parameter:} The modular-width of the input graph $\mw(G)$.
  \item[] \textbf{Question:} Does $S \sevstep_k S'$ hold?
\end{itemize}

\subsection{$\TAR{k}$ rule}

Let $S$ and $S'$ be independent sets in a graph $G$ and $k$ an integer.  Then
we write $S\stackrel{G}{\onestep}_k S'$ if $|\symdiff{S}{S'}|\le 1$ and
$\min\{|S|,|S'|\}\ge k$. If $G$ is clear from the context we simply write
$S\onestep_k S'$. Here $S \onestep_k S'$ means that $S$ and $S'$ can be
reconfigured to each other in one step under the $\TAR{k}$ rule, which stands
for ``Token Addition and Removal'', under the condition that no independent set
contains fewer than $k$ vertices (tokens).  We write
$S\stackrel{G}{\sevstep}_k S'$, or simply $S\sevstep_k S'$ if $G$ is clear, if
there exists $\ell\ge 0$ and a sequence of independent sets $S_0,\ldots,
S_{\ell}$ with $S_0=S$, $S_{\ell}=S'$ and for all $i\in[\ell]$ we have
$S_{i-1}\onestep_k S_i$. If $S\sevstep_k S'$ we say that $S'$ is
\emph{reachable} from $S$ under the $\TAR{k}$ rule.

We recall the following basic facts.
\begin{observation}
\label{obs:basic}
 For all integers $k$ the relation defined by
$\sevstep_k$ is an equivalence relation on independent sets of size at least
$k$. For any graph $G$, integer $k$, and independent sets $S,R$, if
$S\sevstep_k R$, then $S\sevstep_{k-1}R$.  For any graph $G$ and independent
sets $S,R$ we have $S\sevstep_0R$.
\end{observation}

\subsection{$\TJ$ and $\TS$ rules}

Under the $\TJ$ rule, one step is formed by a removal of a vertex and an addition of a vertex.
As this rule does not change the size of the independent set,
we assume that the given initial and target independent sets are of the same size.
In other words, two independent sets $S$ and $S'$ with $|S| = |S'|$
can be reconfigured to each other in one step under the $\TJ$ rule if $|\symdiff{S}{S'}| = 2$.
It is known that the $\TJ$ reachability can be seen as a special case of $\TARrule$ reachability as follows.
\begin{proposition}
[\cite{KaminskiMM12}]
\label{prop:TJ=TAR}
Let $S$ and $R$ be independent sets of $G$ with $|S| = |R|$.
Then, $R$ is reachable from $S$ under $\TJ$
if and only if $S \sevstep_{|S|-1} R$.
\end{proposition}

One step under the $\TS$ rule is a $\TJ$ step
with the additional constraint that the removed and added vertices have to be adjacent.
Intuitively, one step in a $\TS$ sequence ``slides'' a token along an edge.
We postpone the introduction of notation for $\TS$ until Section~\ref{sec:TS} to avoid any confusions.

\subsection{Modular-width}

In a graph $G=(V,E)$ a module is a set of vertices $M\subseteq V$ with the
property that for all $u,v\in M$ and $w\in V\setminus M$, if $\{u,w\}\in E$, then
$\{v,w\}\in E$. In other words, a module is a set of vertices that have the same
neighbors outside the module. A graph $G=(V,E)$ has \emph{modular-width} at
most $k$ if it satisfies at least one of the following conditions (i) $|V|\le
k$, or (ii) there exists a partition of $V$ into at most $k$ sets $V_1,
V_2,\ldots, V_s$, such that $G[V_i]$ has modular-width at most $k$ and $V_i$ is
a module in $G$, for all $i\in[s]$. We will use $\mw(G)$ to denote the minimum
$k$ for which $G$ has modular-width at most $k$. We recall that there is a
polynomial-time algorithm which, given a graph $G=(V,E)$ produces a non-trivial
partition of $V$ into at most $\mw(G)$ modules
\cite{CournierH94,HabibP10,TedderCHP08} and that deleting vertices from $G$ can
only decrease the modular-width. 
We also recall that \textsc{Maximum Independent Set} is solvable in time $O^*(2^{\mw})$.
Indeed, a faster algorithm with running time $O^*(1.7347^{\mw})$ is known~\cite{FominLMT18}.

A graph has neighborhood diversity at most $k$ if its vertex set can be
partitioned into $k$ modules, such that each module induces either a clique or
an independent set. We use $\nd(G)$ to denote the minimum neighborhood
diversity of $G$, and recall that $\nd(G)$ can be computed in polynomial time~\cite{Lampis12}
 and that $\nd(G)\ge \mw(G)$ for all graphs $G$
\cite{GajarskyLO13}.

It can be seen that the modular-width of a graph is not smaller than its clique-width.
On the other hand, we can see that treewidth, pathwidth, and bandwidth are not comparable to modular-width.
To see this, observe that the complete graph of $n$ vertices has treewidth $n-1$ and modular-width $2$,
and that the path of $n$ vertices has treewidth $1$ and modular-width $n$ for $n \ge 4$.
Our positive result and the hardness result by Wrochna~\cite{Wrochna18} together give \figref{fig:width-parameters}
that depicts a map of structural graph parameters with a separation of the complexity of \textsc{Independent Set Reconfiguration}.

\begin{figure}[bth]
  \centering
  \includegraphics[scale=0.75]{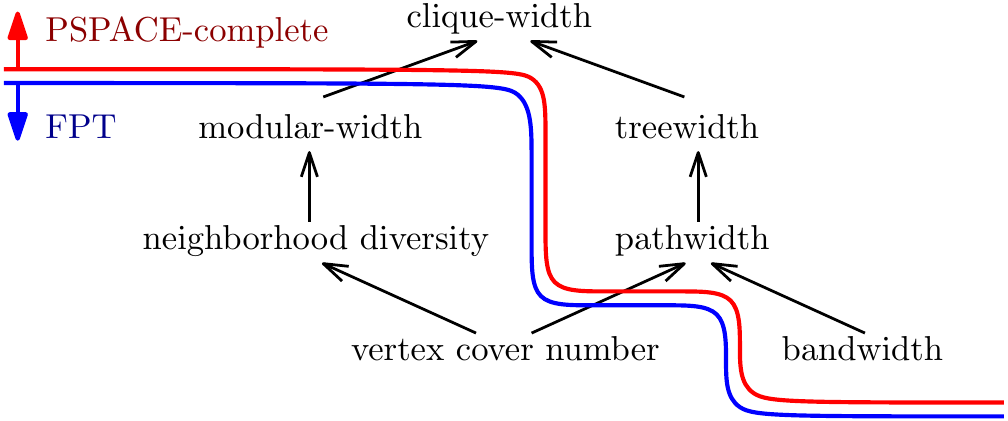}
  \caption{The complexity of \textsc{Independent Set Reconfiguration} under $\TARrule$, $\TJ$, and $\TS$
    parameterized by structural graph parameters. ``$X \rightarrow Y$'' implies that
  there is a function $f$ such that $X(G) \ge f(Y(G))$ for every graph $G$.}
  \label{fig:width-parameters}
\end{figure}


\section{FPT Algorithm for Modular-Width under $\TARrule$}
\label{sec:TAR}

In this section we present an FPT algorithm for the $\TAR{k}$-reachability
problem parameterized by modular-width.  The main technical ingredient of our
algorithm is a sub-routine which solves a related problem: given a graph $G$,
an independent set $S$, and an integer $k$, what is the largest size of an independent set
reachable from $S$ under $\TAR{k}$?  This sub-routine relies on dynamic
programming: we present (in Lemma \ref{lem:lambda}) an algorithm which answers
this ``maximum extensibility'' question, if we are given tables with answers
for the same question for all the modules in a non-trivial partition of the
input graph.  This results in an algorithm (Theorem \ref{thm:lambda}) that
solves this problem on graphs of small modular-width, which we then put to use
in Section \ref{sec:mw2} to solve the reconfiguration problem.

\subsection{Computing a Largest Reachable Set}

In this section we present an FPT algorithm (parameterized by modular-width)
which computes the following value:

\begin{definition}

Given a graph $G$, an independent set $S$, and an integer $k$, we define
$\lambda(G,S,k)$ as the largest size of the independent sets $S'$ such that
$S\sevstep_k S'$.

\end{definition}

In particular, we will present a \emph{constructive} algorithm which, given
$G,S,k$ will return an independent set $S'$ such that $|S'|=\lambda(G,S,k)$, as
well as a reconfiguration sequence proving that $S\sevstep_k S'$.

We begin by tackling an easier case: the case when the parameter is the
neighborhood diversity.

\begin{lemma}
\label{lem:ndlambda} There is an algorithm which, given a graph
$G$, an independent set $S$, and an integer $k$, returns an independent set
$S'$, with $|S'|=\lambda(G,S,k)$, and a reconfiguration sequence proving that
$S\sevstep_k S'$, in time $O^*(2^{\nd(G)})$.  \end{lemma}

\begin{proof}
 Assume that $G=(V,E)$ is partitioned into $r$ sets
$V_1,V_2,\ldots, V_r$ such that each set induces a clique or an independent
set. In fact, we may assume without loss of generality that $V_i$ is an
independent set for each $i\in[r]$, because if $V_i$ is a clique we can delete
all but one of its vertices without affecting the answer.

Consider now the auxiliary graph $G'$ which has a vertex for each independent
set $S$ of $G$ that satisfies the following: (i) $|S|\ge k$ (ii) for all
$i\in[r]$ either $V_i\subseteq S$ or $S\cap V_i=\emptyset$. There are at most
$2^r$ vertices in $G'$. We add an edge between $S_1,S_2$ in $G'$ if
$\symdiff{S_1}{S_2}=V_i$ for some $i\in[r]$. As a result, $G'$ has at most
$r2^r$ edges.

We observe that the set $S'$ we seek is represented by a vertex of $G'$ ($S'$
must be maximal, therefore it fully contains all modules it intersects).
Furthermore, we may assume that the set $S$ we have been given is also
represented by a vertex of $G'$ (because if $0<|S\cap V_i|<|V_i|$ we may add to
$S$ the remaining vertices of $V_i$ and we still have a set that is reachable
from $S$).  We note that $S\sevstep_{k} S'$ if and only if there is a path from $S$
to $S'$ in $G'$, and it is not hard to construct a reconfiguration sequence in
$G$ given a path in $G'$. As a result, the problem reduces to determining the
vertices of $G'$ which are reachable from $S$, and then determining which among
these represents a largest independent set, both of which can be solved in
time linear in the size of $G'$.  
\qed
\end{proof}

Before presenting the main algorithm of this section, let us also make a useful
observation: once we are able to reach a configuration that contains a
sufficiently large number of vertices from a module, we can safely delete a
vertex from the module (bringing us closer to the case where Lemma~\ref{lem:ndlambda} will apply).

\begin{lemma}
\label{lem:delete} Let $G$ be a graph, $S$ be an independent set
of $G$, $k$ an integer, and $M$ a module of $G$. Suppose there exists an
independent set $A\subset M$ such that $(S\cap M)\subseteq A$ and
$|A|=\alpha(G[M])$. Then, for all $u\in M\setminus A$ we have
$\lambda(G,S,k)=\lambda(G-u,S,k)$.
\end{lemma}

\begin{proof}
 We assume that there exists a $u\in M\setminus A$ (otherwise the
claim is vacuously true). We have $u\not\in S$, since $(S\cap M)\subseteq A$
and $u\not\in A$. Therefore, $\lambda(G,S,k)\ge \lambda(G-u,S,k)$, since any
transformation which can be performed in $G-u$ can also be performed in $G$. We
therefore need to argue that $\lambda(G-u,S,k)\ge \lambda(G,S,k)$.

Let $T$ be an independent set of $G$ such that $S\stackrel{G}{\sevstep_k} T$
and $|T|=\lambda(G,S,k)$. Consider a shortest $\TAR{k}$ reconfiguration from
$S$ to $T$, say $S_0=S, S_1, \ldots, S_{\ell}=T$. We construct a $\TAR{k}$
reconfiguration $S_0'=S, S_1', \ldots, S_{\ell}'$ with the property that for
all $i\in[\ell]$ we have $|S_i|=|S_i'|$, $S_i\setminus M=S_i'\setminus M$, and
$(S_i'\cap M)\subseteq A$.  If we achieve this we are done, since we have that
$S\stackrel{G-u}{\sevstep_k} S_{\ell}'$ and $|S_{\ell}'|=|T|=\lambda(G,S,k)$.

We will construct the new reconfiguration sequence inductively. First, $S_0'=S$
satisfies the desired properties. So, suppose that for some $i\in[\ell]$ we
have $|S_{i-1}'|=|S_{i-1}|$, $S_{i-1}\setminus M= S_{i-1}'\setminus M$, and
$(S_{i-1}'\cap M) \subseteq A$.  We now consider the four possible cases
corresponding to single reconfiguration moves from $S_{i-1}$ to $S_i$. If $S_i
= S_{i-1}\setminus\{v\}$, for some $v\in S_{i-1}\setminus M$, we set
$S_i'=S_{i-1}'\setminus\{v\}$; this is a valid move since $v\in S_{i-1}'$.  If
$S_i = S_{i-1}\setminus\{v\}$, for some $v\in S_{i-1}\cap M$, then it must be
the case that $S_{i-1}'\cap M\neq \emptyset$.  Select an arbitrary $v'\in
S_{i-1}'\cap M$ and set $S_i'=S_{i-1}'\setminus \{v'\}$. If $S_i = S_{i-1}\cup
\{v\}$ for $v\in V\setminus M$, we set $S_i' = S_{i-1}'\cup \{v\}$; it is not
hard to see that $S_i'$ is still an independent set and satisfies the desired
properties. Finally, if $S_i = S_{i-1}\cup \{v\}$ for $v\in M$, we observe that
$|S_{i-1}\cap M|<\alpha(G[M])=|A|$ (otherwise $S_i$ would not be independent).
Since $|S_{i-1}\cap M|=|S_{i-1}'\cap M|$ there exists $v'\in A\setminus
S_{i-1}'$. We therefore set $S_i'=S_{i-1}' \cup \{v'\}$. 
\qed
\end{proof}

We are now ready to present our main dynamic programming procedure.

\begin{lemma}
\label{lem:lambda} Suppose we are given the following input:

\begin{enumerate}

\item A graph $G=(V,E)$, an integer $k$, and an independent set $S$ with
$|S|\ge k$.

\item A partition of $V$ into $r\le \mw(G)$ non-empty modules,
$V_1,\ldots,V_r$.

\item For each $i\in[r]$, for each $j\in[ |S\cap V_i| ]$ an independent set
$R_{i,j}$, such that $|R_{i,j}|=\lambda(G[V_i],S\cap V_i,j)$, and a
transformation sequence proving that $(S\cap V_i)\stackrel{G[V_i]}{\sevstep_j}
R_{i,j}$.

\end{enumerate}
Then, there exists an algorithm which returns an independent set $R$ of $G$,
such that $|R|=\lambda(G,S,k)$, and a transformation sequence proving that
$S\sevstep_k R$, running in time $O^*(2^{\mw(G)})$.

\end{lemma}

\begin{proof}
We describe an iterative procedure which maintains the following variables:

\begin{itemize}

\item A working graph $H$, and a working independent set $R$ of $H$.

\item A partition of $V(H)$ into $r+1$ sets $M_0,M_1,\ldots,M_r$, some of which
may be empty.

\item A tuple of $r+1$ non-negative integer ``thresholds'', $t_i$, for
$i\in\{0,\ldots,r\}$.

\end{itemize}

Informally, the meaning of these variables is the following: $H$ is a working
copy of $G$ where we may have deleted some vertices which we have found to be
irrelevant (using Lemma \ref{lem:delete}); $R$ represents a working independent
set which is reachable from the initial set $S$ in $G$ (and we have a
transformation to prove this reachability); $M_0$ represents the union of
initial modules which we have ``processed'' using Lemma \ref{lem:delete}, and
therefore turned into independent sets, which implies that $H[M_0]$ is a graph
with small neighborhood diversity; and $t_i$ represents a threshold above which
we are allowed to perform internal transformations inside the set $M_i$ without
violating the size constraints (that is, while keeping $|R|\ge k$).

To make this more precise we will maintain the following invariants.

\begin{enumerate}

\item $H$ is an induced subgraph of $G$ and $M_0,M_1,\ldots,M_r$ is a partition
of $V(H)$.

\item We have a transformation proving that $S\stackrel{G}{\sevstep_k} R$.

\item $\lambda(G,S,k)= \lambda(H,R,k)$.

\item $\nd(H[M_0])\le |\{i\in[r]\ |\ M_i=\emptyset\}| \le \mw(G)$. 

\item For all $i\in[r]$ we have either $M_i=V_i$ or $M_i=\emptyset$. If
$M_i\neq\emptyset$ then $M_i\cap R\neq \emptyset$.

\item For all $i\in\{0,\ldots,r\}$ such that $M_i\neq \emptyset$ we have
$k-|R\setminus M_i|\le t_i\le |R\cap M_i|$.

\item For each $i\in[r]$ such that $M_i\neq\emptyset$ we have a transformation
$(S\cap V_i)\stackrel{G[V_i]}{\sevstep_{t_i}} (R\cap M_i)$.

\item For  each $i\in[r]$ such that $M_i\neq\emptyset$ we have
$\lambda(G[V_i],S\cap V_i,t_i) = |R\cap M_i|$.

\end{enumerate}

Informally, invariants 1--3 state that we may have deleted some vertices of $G$
and reconfigured the independent set, but this has not changed the answer.
Invariants 4--5 state that $H$ can be thought of as having two parts: the low
neighborhood diversity part induced by $M_0$ and the rest which includes some
unchanged modules of $G$.  Finally for each such non-empty module $M_i$
invariant 6 states that it is safe to perform $\TAR{t_i}$ moves inside $M_i$,
and invariants 7--8 state that the current configuration is reachable and
best possible under such moves, inside the module.

In the remainder, when we say that we ``perform'' a $\TAR{k}$ transformation
from a set $R$ to a set $R'$, what we mean is that our algorithm appends this
transformation to the transformation which we already have from $S$ to $R$ (by
invariant 2), to obtain a transformation from $S$ to the new set $R'$.

\noindent\textbf{Preprocessing}: Our algorithm begins by performing some
preprocessing steps which ensure that all invariants are satisfied. First, for
each $i\in[r]$ we compute a maximum independent set $A_i$ of $G[V_i]$ (this
takes time $O^*(2^{\mw(G)})$).  We initialize $H:=G$ and  $R:=S$.  For each
$i\in[r]$ such that $V_i\cap S=\emptyset$ we set $M_i:=\emptyset$, $t_i:=0$ we
add all vertices of $A_i$ to $M_0$, and we delete from $H$ all vertices of
$M_i\setminus A_i$. After this step we have satisfied invariants 1, 2
(trivially, since $R=S$), 4 (because $M_0$ is composed of the at most $r$
modules which had an empty intersection with $S$, each of which now induces an
independent set), and 5. To see that invariant 3 is satisfied we invoke Lemma
\ref{lem:delete} repeatedly: we see that the lemma trivially applies if $S\cap
V_i=\emptyset$ and allows us to delete all vertices of a module which are
outside a maximum independent set.

For the second preprocessing step we do the following for each $i\in[r]$ for
which $S\cap V_i\neq \emptyset$: we set $M_i:=V_i$ and $t_i=|S\cap V_i|$.  We
recall that we have been given in the input a set $R_{i,j}$ for $j=|S\cap V_i|$
such that $\lambda(G[V_i], S\cap V_i, j)=|R_{i,j}|$, and a transformation
$(S\cap V_i)\stackrel{G[V_i]}{\sevstep_j} R_{i,j}$; we perform this
transformation in $H$, leaving $R\setminus M_i$ unchanged, to obtain an
independent set $R$ such that $R\cap M_i = R_{i,j}$. This is a valid $\TAR{k}$
transformation in $H$ since inside $V_i$ the transformation maintains an
independent set with at least $t_i$ tokens at all times. We observe that after
this step is applied for all $i\in[r]$, all invariants are satisfied.  This
preprocessing step is performed in polynomial time, because the sets $R_{i,j}$
and the transformations leading to them have been given to us as input. Thus,
all our preprocessing steps take a total time of $O^*(2^{\mw(G)})$.

\medskip

For the main part of its execution our algorithm will enter a loop which
attempts to apply some rules (given below) which either delete some vertex of
$H$ or produce a new (larger) independent set $R$, while maintaining all
invariants. Once this is no longer possible the algorithm returns the current
set $R$ as the solution.

\smallskip

\noindent\textbf{Rule 1 (Irrelevant Vertices)}: Check if there exists $i\in[r]$
such that $H[M_i]$ induces at least one edge and $|R\cap M_i| =
\alpha(H[M_i])$. Then, delete all vertices of $M_i\setminus R$ from $H$, set
$M_0:=M_0\cup (M_i\cap R)$, $t_i:=0$, and $M_i:=\emptyset$.

\smallskip

\noindent\textbf{Rule 2a (Configuration Improvement in $M_0$)}: Let
$F_0\subseteq M_0$ be the set of vertices of $M_0$ that have no neighbors in
$\bigcup_{i\in[r]} M_i$.  Let $\lambda_0:= \lambda(H[F_0],R\cap F_0, k-|R\setminus
F_0|)$. If $\lambda_0> |R\cap F_0|$ then set $t_0:=\max\{k-|R\setminus
F_0|,0\}$ and perform a transformation that leaves $R\setminus F_0$ unchanged
and results in $|R\cap F_0|=\lambda_0$.

\noindent\textbf{Rule 2b (Configuration Improvement in $V\setminus M_0$)}: For
each $i\in[r]$ such that $M_i\neq \emptyset$ let
$\lambda_i:=\lambda(H[M_i],R\cap M_i, k-|R\setminus M_i|)$. If there exists $i$
such that $\lambda_i>|R\cap M_i|$, then set $t_i:=\max\{k-|R\setminus M_i|,0\}$
and perform a transformation that leaves $R\setminus M_i$ unchanged and results
in $|R\cap M_i|=\lambda_i$.

\begin{myclaim}\label{claim:1} Rule 1 maintains all invariants and can be applied
in time $O^*(2^{\mw(G)})$. \end{myclaim}

\begin{proof}

First, it is not hard to see that the rule can be applied in the claimed
running time, since the only non-polynomial step is computing $\alpha(H[M_i])$,
which can be done in $O^*(2^{\mw(G)})$.

Let us argue why the rule maintains all invariants. Invariants 1, 2, 5, 6, 7, 8
remain trivially satisfied if they were true before applying the rule. For
invariant 4, we observe that $H[M_0]$ is composed of the (at most $r$) modules
which have been turned into independent sets, therefore its neighborhood
diversity is at most $\mw(G)$. For invariant 3, we invoke Lemma
\ref{lem:delete}. The lemma applies because $R\cap M_i$ is itself a maximum
independent set of $H[M_i]$, and the lemma states that we can safely delete
vertices of $M_i$ outside this independent set.  
\localqed
\end{proof}

\begin{myclaim}\label{claim:2a} Rule 2a maintains all invariants and can be
applied in time $O^*(2^{\mw(G)})$. \end{myclaim}

\begin{proof}

Let us first make the useful observation that $R\cap M_0= R\cap F_0$, which is
a consequence of invariant 5: if $R$ contained a vertex $u\in M_0\setminus F_0$
then this vertex would have a neighbor in some non-empty $M_i$, therefore $u$
would be connected to all of $M_i$ (since $M_i$ is a module), and $R$ would not
be independent since $R\cap M_i\neq \emptyset$.

By invariant 4  we have that $\nd(H[F_0])\le \mw(G)$.  By Lemma
\ref{lem:ndlambda} we can therefore compute $\lambda_0$ in time
$O^*(2^{\mw(G)})$. Furthermore, the algorithm of the lemma returns to us a
transformation $(R\cap F_0)\stackrel{H[F_0]}{\sevstep_{t_0}} R'$, with
$|R'|=\lambda_0$. We perform the same transformation in $H$, keeping
$R\setminus M_0=R\setminus F_0$ unchanged and resulting in $R\cap F_0=R'$. This
is a $\TAR{k}$ transformation in $H$ because by invariant 6, $t_0+|R\setminus
M_0|\ge k$. We therefore maintain invariant 2.  Invariants 1,3,4,5,7,8 are
unaffected by this rule (so remain satisfied).  For invariant 6 observe that we
have set $t_0 := \max\{k-|R\setminus F_0|,0\}\ge k-|R\setminus M_0|$ and that
$R\cap M_0$ increases and $t_0$ does not increase when we apply the rule.
\localqed
\end{proof}

\begin{myclaim}\label{claim:2b} Rule 2b maintains all invariants and can be
applied in polynomial time.  \end{myclaim}

\begin{proof}

First, observe that invariants 1,3,4,5 are unaffected. To ease notation, let
$j:=k-|R\setminus M_i|$. 

Let us first deal with the easy case $j\le 0$. Recall that during the
preprocessing step we have calculated a maximum independent set of $G[V_i]$,
call it $A_i$, and that by invariant 5 $H[M_i]=G[V_i]$. As a result,
$\lambda_i=\lambda(G[V_i],R\cap M_i,0)=|A_i|$ is a known value.  We have set
$t_i=0$, so invariants 6,7 are trivially satisfied no matter how we modify
$R\cap M_i$.  We perform a (trivial) transformation that leaves $R\setminus
M_i$ unchanged and first removes all tokens from $R\cap M_i$ and then adds all
tokens of $A_i$. This is a $\TAR{k}$ transformation in $G$, as $|R\setminus
M_i|\ge k$, so invariant 2 is satisfied, and invariant 8 is satisfied because
$A_i$ is a maximum independent set of $G[V_i]$. We therefore assume in the
remainder that $j>0$.

Let us first argue why $\lambda_i$ can be computed in polynomial time. Again
$H[M_i]=G[V_i]$ by invariant 5, so we want to compute
$\lambda_i=\lambda(G[V_i],R\cap M_i, j)$. By invariant 7, we have a
transformation $(R\cap M_i) \stackrel{G[V_i]}{\sevstep_{t_i}} (S\cap V_i)$ and
by invariant 6 we have $t_i\ge k-|R\setminus M_i|=j$. Therefore, using Observation~\ref{obs:basic}
we have $(R\cap M_i) \stackrel{G[V_i]}{\sevstep}_{j} (S\cap
V_i)$. This means that $\lambda_i=\lambda(G[V_i], S\cap V_i, j)$. Therefore, we
can find the value $\lambda_i$ in the input we have been supplied. Furthermore,
we have in the input an independent set $R_{i,j}$ that is $\TAR{j}$ reachable
from $S\cap V_i$ in $G[V_i]$ and has $|R_{i,j}|=\lambda_i$; this set is also
$\TAR{j}$ reachable from $R\cap M_i$, so we can perform a transformation that
leaves $R\setminus M_i$ unchanged and results in $R\cap M_i = R_{i,j}$. This
maintains invariant 2, as well as invariants 7 and 8. Finally, invariant 6 is
satisfied by the new value of $t_i$ since $t_i$ may only decrease and $R\cap
M_i$ increases.  
\localqed
\end{proof}

We are now ready to argue for the algorithm's correctness. First, we observe
that by Claims \ref{claim:1}, \ref{claim:2a}, \ref{claim:2b}, all rules can be
applied in time at most $O^*(2^{\mw(G)})$. Since all rules either delete a
vertex of the graph or increase $|R|$, the algorithm runs in time
$O^*(2^{\mw(G)})$. By invariant 2, the independent set $R$ returned by the
algorithm (when no rule can be applied) is $\TAR{k}$ reachable from $S$ in $G$,
therefore $\lambda(G,S,k)\ge |R|$. We therefore need to argue that $|R|\ge
\lambda(G,S,k)$. By invariant 3, this is equivalent to $|R|\ge \lambda(H,R,k)$.
In other words, we want to argue that in the graph $H$ on which the algorithm
may no longer apply any rules, it is impossible to reach an independent set $T$
using $\TAR{k}$ moves such that $|T|>|R|$.

Let $M_0,M_1,\ldots, M_r$ be the partition of $V(H)$ at the last iteration of
our algorithm. For an independent set $T$ of $H$ we will say that $T$ is
\emph{interesting} if it satisfies at least one of the following two
conditions: (a) for some $i\in\{0,\ldots,r\}$ we have $|T\cap M_i|> |R\cap
M_i|$ or (b) for some $i\in[r]$ we have $R\cap M_i\neq \emptyset$ and $T\cap
M_i=\emptyset$. In other words, an independent set is interesting if it manages
to place more tokens than $R$ in a module $M_i$ (or in $M_0$), or if it manages
to remove all tokens from a module $M_i$ that is non-empty in $R$.

We will make two claims: (i) no interesting set is reachable from $R$ with
$\TAR{k}$ moves in $H$; (ii) if for an independent set $T$ of $H$ we have
$|T|>|R|$, then $T$ is interesting.  Together the two claims imply that $|R|\ge
\lambda(H,R,k)$, since all sets which are strictly larger than $R$ are not
reachable.

For claim (ii) suppose that $T$ is not interesting, therefore
$|T\cap M_i|\le |R\cap M_i|$ for each $i\in \{0,\ldots,r\}$.
Since $M_0,M_1,\ldots, M_r$ is a partition of $V(H)$ we have $|T|\le |R|$.

For claim (i) suppose that $T$ is interesting and $R\sevstep_k T$. Among all
such sets $T$ consider one whose shortest reconfiguration sequence from $R$
has minimum length. Let $R_0=R, R_1,\ldots, R_{\ell}=T$ be such a shortest
reconfiguration sequence. Therefore, we have $\ell\ge 1$ and $R_0,\ldots,
R_{\ell-1}$ are not interesting. We now consider the two possible reasons for
which $R_{\ell}$ may be interesting.

In case there exists $i\in\{0,\ldots,r\}$ such that $|R_{\ell}\cap M_i|> |R\cap
M_i|$ we construct a transformation from $R\cap M_i$ to $R_{\ell}\cap M_i$ in
$H[M_i]$ by considering the sets $R_0\cap M_i, R_1\cap M_i, \ldots,
R_{\ell}\cap M_i$.  We note that this is a $\TAR{k-|R\setminus M_i|}$
transformation, since for all $j<\ell$ we have $|R_j\setminus M_i|\le
|R\setminus M_i|$ and $|R_j|\ge k$. But then Rule 2 could have been applied. In
particular, if $i\neq 0$ the transformation proves that
$\lambda_i=\lambda(H[M_i], R\cap M_i, k-|R\setminus M_i|) \ge |R_{\ell}\cap
M_i| > |R\cap M_i|$ so we could have applied Rule 2b.  If $i=0$ we first make
the observation that for all $j\in\{0,\ldots,\ell\}$ we have $R_j\cap M_0 =
R_j\cap F_0$ (where $F_0$ is defined in Rule 2a). To see this we argue (as in
Claim \ref{claim:2a}) that if there exists $u\in R_j \cap (M_0\setminus F_0)$
then $u$ is connected to all of a module $M_i$ which has $R\cap M_i\neq
\emptyset$, and therefore $R_j\cap M_i\neq \emptyset$ (as $R_j$ is not
interesting (does not satisfy (b)) for $j<\ell$) which contradicts the independence of $R_j$. We now
have $\lambda_0=\lambda(H[F_0], R\cap F_0, k-|R\setminus M_0|) \ge
|R_{\ell}\cap F_0| > |R\cap F_0|$ and Rule 2a could have been applied.

In the case there exists $i\in[r]$ such that $R\cap M_i\neq\emptyset$ and
$R_{\ell}\cap M_i=\emptyset$, this means that $|R_{\ell-1}\setminus M_i|\ge k$,
which implies that $|R\setminus M_i|\ge k$, because $R_{\ell-1}$ is not
interesting (does not satisfy (a)). If $|R\cap M_i|=\alpha(G[V_i])$, then Rule 1 would have been
applied, so we assume $|R\cap M_i|<\alpha(G[V_i])$. However, $k-|R\setminus
M_i|\le 0$, so we have $\lambda_i=\lambda(H[M_i],R\cap M_i, k-|R\setminus
M_i|)=\alpha(G[V_i]) > |R\setminus M_i|$, therefore Rule 2b should have been
applied.

Thus, in both cases we see that a rule could have been applied and we have a
contradiction.  Therefore, the set $R$ returned is optimal.  
\qed
\end{proof}

We thus arrive to the main theorem of this section.

\begin{theorem}\label{thm:lambda}

There exists an algorithm which, given a graph $G$, an independent set $S$, and
an integer $k$, runs in time $O^*(2^{\mw(G)})$ and outputs an independent set
$S'$ such that $|S'|=\lambda(G,S,k)$ and a $\TAR{k}$ transformation
$S\sevstep_k S'$.

\end{theorem}

\begin{proof}

We perform dynamic programming using Lemma \ref{lem:lambda}. More precisely,
our goal is, given $G$ and $S$, to produce for each value of $j\in [|S|]$ an
independent set $R_{j}$ such that $S\sevstep_{j} R_j$ and $|R_j|=\lambda(G,S,j)$.
Clearly, if we can solve this more general problem in time $O^*(2^{\mw(G)})$ we
are done.

Our algorithm works as follows: first, it computes a modular decomposition of
$G$ of minimum width, which can be done in time at most $O(n^2)$
\cite{HabibP10}. If $|V(G)|\le \mw(G)$ then the problem can be solved in
$O^*(2^{\mw(G)})$ by brute force (enumerating all independent sets of $G$), or
even by Lemma \ref{lem:ndlambda}. We therefore assume that $G$ has a
non-trivial partition into $r\le\mw(G)$ modules $V_1,\ldots, V_r$. We call our
algorithm recursively for each $G[V_i]$, and obtain for each $i\in[r]$ and
$j\in [|S\cap V_i|]$ a set $R_{i,j}$ such that $|R_{i,j}|=\lambda(G[V_i],S\cap
V_i,j)$ and a transformation $(S\cap V_i)\stackrel{G[V_i]}{\sevstep_j}
R_{i,j}$. We use this input to invoke the algorithm of Lemma \ref{lem:lambda}
for each value of $j\in[|S|]$. This allows us to produce the sets $R_j$ and the
corresponding transformations.

Suppose that $\beta\ge 2$ is a constant such that the algorithm of Lemma
\ref{lem:lambda} runs in time at most $O(2^{\mw(G)}n^{\beta})$. Our algorithm
runs in time at most $O(2^{\mw(G)}n^{\beta+2})$. This can be seen by
considering the tree representing a modular decomposition of $G$. In each node
of the tree (that represents a module of $G$) our algorithm makes at most $n$
calls to the algorithm of Lemma \ref{lem:lambda}. Since the modular
decomposition has at most $O(n)$ nodes, the running time bound follows.
\qed
\end{proof}

\subsection{Reachability}\label{sec:mw2}

In this section we will apply the algorithm of Theorem \ref{thm:lambda} to
obtain an FPT algorithm for the $\TAR{k}$ reconfiguration problem parameterized
by modular-width. The main ideas we will need are that (i) using the algorithm
of Theorem~\ref{thm:lambda} we can decide if it is possible to arrive at a
configuration where a module is empty of tokens (Lemma \ref{lem:empty}) (ii) if
a module is empty in both the initial and target configurations, we can replace
it by an independent set (Lemma \ref{lem:delete2}) and (iii) the
reconfiguration problem is easy on graphs with small neighborhood diversity
(Lemma \ref{lem:ndreach}).  Putting these ideas together we obtain an algorithm
which can always identify an irrelevant vertex which we can delete if the input
graph is connected. If the graph is disconnected, we can use ideas similar to
those of \cite{Bonsma16} to reduce the problem to appropriate sub-instances in
each component.

\begin{lemma}\label{lem:empty} There is an algorithm which, given a graph $G$,
an independent set $S$, a module $M$ of $G$, and an integer $k$, runs in time
$O^*(2^{\mw(G)})$ and either returns a set $S'$ with $S'\cap M=\emptyset$ and
$S\sevstep_k S'$ or correctly concludes that no such set exists. \end{lemma}

\begin{proof}

We assume that $S\cap M\neq \emptyset$ (otherwise we simply return $S$). 

Let $H$ be the graph obtained by deleting from $G$ all vertices of $V\setminus
M$ that have a neighbor in $M$. We invoke the algorithm of Theorem
\ref{thm:lambda} to compute a set $R$ in $H$ such that
$S\stackrel{H}{\sevstep_k}R$ and $|R|=\lambda(H,S,k)$. If $|R\setminus M|\ge
k$ then we return as solution the set $R\setminus M$, and as transformation
the transformation sequence returned by the algorithm, to which we append moves
that delete all vertices of $R\cap M$. If $|R\setminus M|<k$ we answer that
no such set exists.

Let us now argue for correctness. If the algorithm returns a set $S' := R\setminus M$, it also
returns a $\TAR{k}$ transformation from $S$ to $S'$ in $H$; this is also a
transformation in $G$, and since $S'\cap M=\emptyset$, the solution is
correct. 

Suppose then that the algorithm returns that no solution exists, but for the
sake of contradiction there exists a $T$ with $S\stackrel{G}{\sevstep_k}T$ and
$T\cap M=\emptyset$. Among all such sets $T$ select the one at minimum
reconfiguration distance from $S$ and let $S_0=S, S_1,\ldots, S_{\ell}=T$ be a
shortest reconfiguration sequence. We claim that this is also a valid
reconfiguration sequence in $H$. Indeed, for all $j\in[\ell-1]$, the set $S_j$
contains a vertex from $M$ (otherwise we would have a shorter sequence),
therefore may not contain any deleted vertex.  As a result, if a solution $T$
exists, then $S\stackrel{H}{\sevstep_k}T$. Let $A$ be a maximum independent set
of $G[M]$.  We observe that (i) $|T\setminus M|\ge k$ since $T$ is reachable
with $\TAR{k}$ moves and $T\cap M=\emptyset$ (ii) $T\stackrel{H}{\sevstep_k}
(T\cup A)$.  However, this gives a contradiction, because we now have
$S\stackrel{H}{\sevstep_k} (T\cup A)$ and this set is strictly larger than the
set returned by the algorithm of Theorem \ref{thm:lambda} when computing
$\lambda(H,S,k)$.  
\qed
\end{proof}

\begin{lemma}\label{lem:delete2} Let $G$ be a graph, $k$ an integer, $M$ a
module of $G$, and $S,T$ two independent sets of $G$ such that $S\cap M = T\cap
M=\emptyset$.  Let $A$ be a maximum independent set of $G[M]$. Then, for all
$u\in M\setminus A$ we have $S\stackrel{G}{\sevstep_k}T$ if and only if
$S\stackrel{G-u}{\sevstep_k}T$.\end{lemma}

\begin{proof}

The proof is similar to that of Lemma \ref{lem:delete}. Specifically, since
$u\not\in S$ and $u\not\in T$, it is easy to see that
$S\stackrel{G-u}{\sevstep_k}T$ implies $S\stackrel{G}{\sevstep_k}T$. Suppose
then that $S\stackrel{G}{\sevstep_k}T$ and we have a sequence $S_0=S,S_1,\ldots, S_{\ell}=T$.
We construct a sequence $S_0'=S, S_1', \ldots,S_{\ell}'$ 
such that for all $i\in[\ell]$ we have $|S_i|=|S_i'|$, $S_i\setminus M = S_i'\setminus M$,
and $S_i' \cap M \subseteq A$. This can be done inductively: for
$S_0'$ the desired properties hold; and for all $i\in[\ell]$ we can prove that
if the properties hold for $S_{i-1}'$, then we can construct $S_i'$ in the same
way as in the proof of Lemma \ref{lem:delete} (namely, we perform the same
moves as $S_i$ outside of $M$, and pick an arbitrary vertex of $A$ when $S_i$
adds a vertex of $M$). 
\qed
\end{proof}

\begin{lemma}\label{lem:ndreach} There is an algorithm which, given a graph
$G$, an integer $k$, and two independent sets $S,T$, decides if $S\sevstep_k T$
in time $O^*(2^{\nd(G)})$.  \end{lemma}

\begin{proof} The proof is similar to that of Lemma \ref{lem:ndlambda}, but we
need to carefully handle some corner cases. We are given a partition of $G$
into $r\le\nd(G)$ sets $V_1,\ldots, V_r$, such that each $V_i$ induces a clique
or an independent set. Suppose $V_i$ induces a clique. We use the algorithm of
Lemma \ref{lem:empty} with input ($G$,$S$,$V_i$,$k$) and with input
($G$,$T$,$V_i$,$k$) to decide if it is possible to empty $V_i$ of tokens. If
the algorithm gives different answers we immediately reject, since there is a
configuration that is reachable from $S$ but not from $T$. If the algorithm
returns $S',T'$ with $S'\cap V_i=T'\cap V_i=\emptyset$, then the problem
reduces to deciding if $S'\sevstep_k T'$. However, by Lemma \ref{lem:delete2}
we can delete all the vertices of $V_i$ except one and this does not change the
answer. Finally, if the algorithm responds that $V_i$ cannot be empty in any
configuration reachable from $S$ or $T$ then, if $S\cap V_i\neq T\cap V_i$ we
immediately reject, while if $S\cap V_i = T\cap V_i$ we delete from the input
$V_i$ and all its neighbors and solve the reconfiguration problem in the
instance ($G[V\setminus N[V_i]]$, $k-1$, $S\setminus V_i$, $T\setminus V_i$).

After this preprocessing all sets $V_i$ are independent. We now construct an
auxiliary graph $G'$ as in Lemma \ref{lem:ndlambda}, namely, our graph has a
vertex for every independent set $S$ of $G$ with $|S|\ge k$ such that for all
$i\in[r]$ either $S\cap V_i=\emptyset$ or $V_i\subseteq S$. Again, we have an
edge between $S_1,S_2$ if $\symdiff{S_1}{S_2}=V_i$ for some $i\in[r]$. We can
assume without loss of generality that $S,T$ are represented in this graph (if
there exists $V_i$ such that $0<|S\cap V_i|<|V_i|$ we add to $S$ all remaining
vertices of $V_i$).  Now, $S\sevstep T$ if and only if $S$ is reachable from
$T$ in $G'$, and this can be checked in time linear in the size of $G'$.
\qed
\end{proof}

\begin{theorem}
[$\TARrule$] 
\label{thm:TAR}
There is an algorithm which, given a graph $G$, an integer $k$,
and two independent sets $S,T$, decides if $S\sevstep_k T$ in time
$O^*(2^{\mw(G)})$.  
\end{theorem}

\begin{proof}
Our algorithm considers two cases: if $G$ is connected we will attempt to
simplify $G$ in a way that eventually produces either a graph with small
neighborhood diversity or a disconnected graph; if $G$ is disconnected we will
recursively solve an appropriate subproblem in each component.

First, suppose that $G$ is connected. We compute a modular decomposition of $G$
which gives us a partition of $V$ into $r\le\mw(G)$ modules $V_1,\ldots, V_r$.
We may assume that $r \ge 2$ since otherwise $G$ has at most $\mw(G)$
vertices and the claimed running time is trivial in that case.
If for all $i\in[r]$ we have that $G[V_i]$ is an independent set, then
$\nd(G)\le r$ and we invoke the algorithm of Lemma \ref{lem:ndreach}. Suppose
then that for some $i\in [r]$, $G[V_i]$ contains at least one edge. We invoke
the algorithm of Lemma \ref{lem:empty} on input $(G,S,V_i,k)$ and on input
$(G,T,V_i,k)$. If the answers returned are different, we decide that $S$ is not
reachable from $T$ in $G$, because from one set we can reach a configuration
that contains no vertex of $V_i$ and from the other we cannot. 

If the algorithm of Lemma \ref{lem:empty} returned to us two sets $S',T'$ with
$S'\cap V_i = T'\cap V_i = \emptyset$ then by transitivity we know $S\sevstep
T$ if and only if $S'\sevstep T'$. We compute a maximum independent set $A$ of
$G[V_i]$ and delete from our graph a vertex $u\in V_i\setminus A$. Such a
vertex exists, since $G[V_i]$ is not an independent set. By Lemma
\ref{lem:delete2} deleting $u$ does not affect whether $S'\sevstep T'$, so we
call our algorithm with input $(G-u,k,S',T')$, and return its response. 

On the other hand, if the algorithm of Lemma \ref{lem:empty} concluded that no
set reachable from either $S$ or $T$ has empty intersection with $V_i$, we find
a vertex $u\in V\setminus V_i$ that has a neighbor in $V_i$ and delete it, that
is, we call our algorithm with input $(G-u,k,S,T)$. Such a vertex $u$ exists
because $G$ is connected. This recursive call is correct because any
configuration reachable from $S$ or $T$ contains some vertex of $V_i$, which is
a neighbor of $u$, so no reachable configuration uses $u$.

We note that if $G$ is connected, all the cases described above will make a
single recursive call on an input that has strictly fewer vertices.

Suppose now that $G$ is not connected and there are $s$ connected components
$C_1,C_2,\ldots, C_s$. We will assume that $|S|=\lambda(G,S,k)$ and
$|T|=\lambda(G,T,k)$. This is without loss of generality, since we can invoke
the algorithm of Theorem \ref{thm:lambda} and in case $|S|<\lambda(G,S,k)$
replace $S$ with the set $S'$ returned by the algorithm while keeping an
equivalent instance (similarly for $T$).

As a result, we can assume that $|S|=|T|$, otherwise the answer is trivially
no. More strongly, if there exists a component $C_i$ such that $|S\cap C_i|\neq
|T\cap C_i|$ we answer no. To see that this is correct, we argue that for all
$S'$ such that $S\sevstep_k S'$ we have $|S'\cap C_i|\le |S\cap C_i|$. Indeed,
suppose there exists $S'$ such that for some $i\in[s]$ we have $|S'\cap
C_i|>|S\cap C_i|$ and $S\sevstep S'$. Among such configurations $S'$ select one
that is at minimum reconfiguration distance from $S$ and let $S_0=S,S_1,\ldots,
S_{\ell}=S'$ be a shortest reconfiguration from $S$ to $S'$.  Then for all
$j\in[\ell]$ we have $|S\setminus C_i|\ge |S_j\setminus C_i|$ (otherwise we
would have an $S'$ that is at shorter reconfiguration distance from $S$). This
means that the sequence $S_0\cap C_i, S_1\cap C_i,\ldots, S_{\ell}\cap C_i$ is
a $\TAR{k-|S\setminus C_i|}$ transformation of $S\cap C_i$ to $S'\cap C_i$ in
$G[C_i]$. But this transformation proves that the set $(S\setminus C_i)\cup
(S'\cap C_i)$ is $\TAR{k}$ reachable from $S$ in $G$, and since this set is
larger than $S$ we have a contradiction.

For each $i\in[s]$ we now consider the reconfiguration instance given by the
following input: $(G[C_i],k-|S\setminus C_i|,S\cap C_i, T\cap C_i)$. We call
our algorithm recursively for each such instance. If the answer is yes for all
these instances we reply that $S$ is reachable from $T$, otherwise we reply
that the sets are not reachable.

To argue for correctness we use induction on the depth of the recursion.
Suppose that the algorithm correctly concludes that the answer to all
sub-instances is yes. Then, there does indeed exist a transformation $S\sevstep
T$ as follows: starting from $S$, for each $i\in[s]$ we keep $S\setminus C_i$
constant and perform in $G[C_i]$ the transformation $(S\cap
C_i)\stackrel{G[C_i]}{\sevstep}_{k-|S\setminus C_i|} (T\cap C_i)$. At each step
this gives a configuration where $S$ and $T$ agree in more components.
Furthermore, since $|S\cap C_i|=|T\cap C_i|$ for all $i\in[s]$, this is a valid
$\TAR{k}$ reconfiguration.

Suppose now that the answer is no for the instance $(G[C_i],k-|S\setminus C_i|,
S\cap C_i, T\cap C_i)$.  Suppose also, for the sake of contradiction, that
there exists a $\TAR{k}$ reconfiguration $S_0=S, S_1,\ldots, S_{\ell}=T$. As
argued above, any configuration $S'$ reachable from $S$ has $|S'\cap C_{i'}|\le
|S\cap C_{i'}|$ for all $i'\in[s]$. This means that $|S\setminus C_i|\ge
|S_j\setminus C_i|$ for all $j\in[\ell]$. Hence, the sequence $S_0\cap C_i,
S_1\cap C_i,\ldots, S_{\ell}\cap C_i$ gives a valid $\TAR{k-|S\setminus C_i|}$
reconfiguration in $G[C_i]$, which is a contradiction.

Finally, it is not hard to see that the algorithm runs in time
$O^*(2^{\mw(G)})$, because in the case of disconnected graphs we make a single
recursive call for each component.  
\qed
\end{proof}

Theorem~\ref{thm:TAR} and Proposition~\ref{prop:TJ=TAR} give an FPT algorithm with the same running time for $\TJ$.
\begin{corollary}
[$\TJ$]
\label{cor:TJ}
There is an algorithm which, given a graph $G$ and two independent sets $S,T$,
decides the $\TJ$ reachability between $S$ and $T$ in time $O^*(2^{\mw(G)})$.  
\end{corollary}


\section{FPT Algorithm for Modular-Width under $\TS$}
\label{sec:TS}

We now present an FPT algorithm deciding the $\TS$-reachability parameterized by modular-width.
The problem under $\TS$ is much easier than the one under $\TARrule$
since we can reduce the problem to a number of constant-size instances that can be considered separately.
To see this, we first observe that the components can be considered separately.
We then further observe that we only need to solve the case
where each maximal nontrivial module contains at most one vertex of the current independent set.
Finally, we show that the reachability problem on the reduced case thus far is 
equivalent to a generalized reachability problem on a graph of order at most $\mw(G)$, where $G$ is the original graph.

Let $S$ and $S'$ be independent sets of $G$ with $|S| = |S'|$.
Recall that $S$ and $S'$ can be reached by one step under $\TS$ if 
$|\symdiff{S}{S'}| = 2$ and the two vertices in $\symdiff{S}{S'}$ are adjacent.
We denote this relation by $S \stackrel{G}{\onestep} S'$, or simply by 
$S \onestep S'$ if $G$ is clear from the context.
We write $S\stackrel{G}{\sevstep} S'$ (or simply $S \sevstep S'$) 
if there exists $\ell \ge 0$ and a sequence of independent sets $S_0,\dots,S_{\ell}$
 with $S_0=S$, $S_{\ell}=S'$ and for all $i\in[\ell]$ we have
$S_{i-1}\onestep S_i$. If $S\sevstep S'$ we say that $S'$ is
\emph{reachable} from $S$ under the $\TS$ rule.
Observe that the relation defined by $\sevstep$ is an equivalence relation on independent sets.

The first easy observation is that the $\TS$ rule cannot move a token to a different component
since a $\TS$ step is always along an edge (and thus within a component).
This is formalized as follows.
\begin{observation}
\label{obs:TS-component}
Let $G$ be a graph, $S,S'$ independent sets of $G$,
and $C_{1},\dots,C_{c}$ the components of $G$.
Then, $S \stackrel{G}{\sevstep} S'$ if and only if
$(S \cap V(C_{i})) \stackrel{G[V(C_{i})]}{\sevstep} (S' \cap V(C_{i}))$ for all $i \in [c]$.
\end{observation}

The next lemma, which is still an easy one, is a key tool in our algorithm.
\begin{lemma}
\label{lem:TS-at-least-2}
Let $G$ be a graph, $M$ a module of $G$, and $S$ an independent set in $G$ such that $|S \cap M| \ge 2$.
Then, for every independent set $S'$ in $G$, $S \stackrel{G}{\sevstep} S'$ if and only if
$S' \cap N(M) = \emptyset$ and $S \stackrel{G - N(M)}{\sevstep} S'$.
\end{lemma}

\begin{proof}
The if direction holds as a $\TS$ sequence in an induced subgraph is always valid in the original graph.

To show the only-if direction, assume that $S \stackrel{G}{\sevstep} S'$.
Let $S_{0}, \dots, S_{\ell}$ be a $\TS$ sequence from $S_{0} = S$ to $S_{\ell} = S'$.
It suffices to show that no independent set in the sequence contains a vertex in $N(M)$.
Suppose to the contrary that $i$ is the first index such that $S_{i} \cap N(M) \ne \emptyset$.
Since $|S \cap M| \ge 2$, we have $S_{i} \cap M \ne \emptyset$.
As $M$ is a module, all vertices in $M$ are adjacent to all vertices in $N(M)$.
This contradicts that $S_{i}$ is an independent set.
Therefore, $S_{i} \cap N(M) = \emptyset$ for all $i$.
\qed
\end{proof}

Lemma~\ref{lem:TS-at-least-2} implies that $S$ with $|S \cap M| \ge 2$ and $S'$ with $|S' \cap M| \le 1$
are not reachable to each other.
This fact in the following form will be useful later.
\begin{corollary}
\label{cor:TS-at-most-1}
Let $G$ be a graph, $M$ a module of $G$, and $S$ an independent set in $G$ such that $|S \cap M| \le 1$.
Then, for every independent set $S'$ in $G$ such that $S \sevstep S'$, it holds that $|S' \cap M|\le 1$.
\end{corollary}

We now show that a module sharing at most one vertex with both initial and target independent sets can be replaced with a single vertex,
under an assumption that we may solve a slightly generalized reachability problem (which is still trivial on a graph of constant size).

\begin{lemma}
\label{lem:TS-at-most-1}
Let $G$ be a graph, $M$ a module of $G$ with $|M| \ge 2$, and $S, S'$ independent sets of $G$ with $|S| = |S'|$.
If $|M \cap (S \cup S')| \le 1$, then $S \stackrel{G}{\sevstep} S'$ if and only if $S \stackrel{G - v}{\sevstep} S'$
for every $v \in M \setminus (S \cup S')$.
\end{lemma}
\begin{proof}
If $S \stackrel{G - v}{\sevstep} S'$, then $S \stackrel{G}{\sevstep} S'$
since $G - v$ is an induced subgraph of $G$.

To prove the only-if part, assume that $S \stackrel{G}{\sevstep} S'$.
Let $S_{0}, \dots, S_{\ell}$ be a $\TS$ sequence from $S_{0} = S$ to $S_{\ell} = S'$.
If $M \cap (S \cup S')$ is nonempty,
let $u$ be the unique vertex in $M \cap (S \cup S')$;
otherwise, let $u$ be an arbitrary vertex in $M - v$.
For $0 \le i \le \ell$, we define $T_{i}$ as follows:
if $S_{i} \cap M = \emptyset$, then $T_{i} = S_{i}$; 
otherwise, $T_{i} = (S_{i} \setminus M) + u$.
Observe that $T_{0} = S$ and $T_{\ell} = S'$.
We show that $T_{i-1} \stackrel{G-v}{\sevstep} T_{i}$ for each $i \in [\ell]$.
By Corollary~\ref{cor:TS-at-most-1}, $|M \cap S_{i}| \le 1$ holds for all $i$.
This implies that $|T_{i}| = |S_{i}|$
and that $T_{i}$ is an independent set of $G-v$ since $N_{G}(w) \setminus M = N_{G}(u) \setminus M$ for all $w \in M$.

Let $S_{i-1} \setminus S_{i} = \{x\}$ and $S_{i} \setminus S_{i-1} = \{y\}$.
Note that $\{x,y\} \in E(G)$.
If $x, y \in M$, then $T_{i-1} = S_{i-1} - x + u = S_{i} - y + u = T_{i}$,
and thus $T_{i-1} \stackrel{G-v}{\sevstep} T_{i}$.
If $x, y \notin M$, then $T_{i-1} \setminus T_{i} = \{x\}$
and $T_{i} \setminus T_{i-1} = \{y\}$.
Since $\{x,y\} \in E(G-v)$, we have $T_{i-1} \stackrel{G-v}{\onestep} T_{i}$.
In the remaining case, we may assume by symmetry that $x \in M$ and $y \notin M$.
Since $\{x,y\} \in E(G)$, $u,x \in M$, and $v \notin \{u,y\}$, we have $\{u,y\} \in E(G-v)$.
Since $T_{i-1} \setminus T_{i} = \{u\}$ and $T_{i} \setminus T_{i-1} = \{y\}$,
we have $T_{i-1} \stackrel{G-v}{\onestep} T_{i}$.
\qed
\end{proof}

\begin{lemma}
\label{lem:TS-both1}
Let $G$ be a graph, $M$ a module of $G$, and $S,S'$ independent sets of $G$ with $|S| = |S'|$.
If $M \cap S = \{u\}$, $M \cap S' = \{v\}$,  $u \ne v$, and $u$ and $v$ are in the same component of $G[M]$,
then $S \stackrel{G}{\sevstep} S'$ if and only if $S \stackrel{G-v}{\sevstep} S'-v+u$.
\end{lemma}
\begin{proof}
Let $P = (p_{0}, \dots, p_{q})$ be a $u$--$v$ path in $G[M]$, where $p_{0} = u$ and $p_{q} = v$.
Let $T_{0} = S'-v+u$ and $T_{i} = T_{i-1} - p_{i-1} + p_{i}$ for $i \ge 1$.
Clearly, $|\symdiff{T_{i-1}}{T_{i}}| = 2$ for $i \in [q]$.
Each $T_{i}$ is an independent set of $G$ since $T_{0}$ is an independent set of $G$,
$T_{i} = T_{0} - u + p_{i}$, and $N_{G}(u) \setminus M = N_{G}(p_{i}) \setminus M$ as $u$ and $p_{i}$ are in the same module $M$.
Since $\symdiff{T_{i-1}}{T_{i}} = \{p_{i-1},p_{i}\}$ and $\{p_{i-1},p_{i}\} \in E(G)$ for each $i \in [q]$,
we have $T_{i-1} \stackrel{G}{\onestep} T_{i}$.
As $T_{0} = S'-v+u$ and $T_{q} = S'$, we have $S'-v+u \stackrel{G}{\sevstep} S'$.
Hence, we can conclude that $S \stackrel{G}{\sevstep} S'$ if and only if $S \stackrel{G}{\sevstep} S'-v+u$.
On the other hand, since $|M| \ge 2$ and $M \cap (S \cup (S'-v+u))) = \{u\}$,
Lemma~\ref{lem:TS-at-most-1} implies that
$S \stackrel{G}{\sevstep} S'-v+u$ if and only if $S \stackrel{G-v}{\sevstep} S'-v+u$.
Putting them together, 
we obtain that $S \stackrel{G}{\sevstep} S'$ if and only if $S \stackrel{G-v}{\sevstep} S'-v+u$. 
\qed
\end{proof}

\begin{lemma}
\label{lem:TS-both2}
Let $G$ be a graph, $M$ a module of $G$, and $S, S'$ independent sets of $G$ with $|S| = |S'|$.
If $M \cap S = \{u\}$, $M \cap S' = \{v\}$,  $u \ne v$, and $u$ and $v$ are in different components of $G[M]$,
then $S \stackrel{G}{\sevstep} S'$ if and only if 
$S \stackrel{G-v}{\sevstep} S'-v+u$
and there is an independent set $T$ in $G-v$
such that $T \cap M = \emptyset$ and $S \stackrel{G-v}{\sevstep} T$.
\end{lemma}
\begin{proof}
We first show the only-if part. Assume that $S \stackrel{G}{\sevstep} S'$.
Let $S_{0}, \dots, S_{\ell}$ be a $\TS$ sequence from $S_{0} = S$ to $S_{\ell} = S'$.
Observe that there is an index $i$ such that $S_{i} \cap M = \emptyset$.
This is because, otherwise the vertex in $S_{i} \cap M$ is always in the same component of $G[M]$ that contains $u$ and thus cannot reach $v$.
We set $T = S_{i}$.
By Lemma~\ref{lem:TS-at-most-1}, we have $S \stackrel{G-v}{\sevstep} T$.
We also have $T \stackrel{G-v}{\sevstep} S-v+u$ by Lemma~\ref{lem:TS-at-most-1},
and hence $S \stackrel{G-v}{\sevstep} S-v+u$.

To show the if part, assume that $S \stackrel{G-v}{\sevstep} S'-v+u$
and $S \stackrel{G-v}{\sevstep} T$ for some independent set $T$ in $G-v$ with $T \cap M = \emptyset$.
These assumptions imply that $S \stackrel{G}{\sevstep} T$ and $T \stackrel{G}{\sevstep} S'-v+u$.
Let $S_{0}, \dots, S_{\ell}$ be a $\TS$ sequence from $S_{0} = T$ to $S_{\ell} = S'-v+u$ in $G$.
For each $i$, we set $T_{i} = S_{i}$ if $S_{i} \cap M = \emptyset$; otherwise, we set $T_{i} = (S_{i} \setminus M) + v$.
Note that $T_{0} = T$ and $T_{\ell} = S'$. We show that $T_{i-1} \stackrel{G}{\sevstep} T_{i}$ for each $i \in [\ell]$.
Corollary~\ref{cor:TS-at-most-1} implies that $|M \cap S_{i}| \le 1$ holds for all $i$, and thus $|T_{i}| = |S_{i}|$.
We can see that $T_{i}$ is an independent set of $G$
since $N_{G}(v) \setminus M = N_{G}(w) \setminus M$ for all $w \in M$.

Let $S_{i-1} \setminus S_{i} = \{x\}$ and $S_{i} \setminus S_{i-1} = \{y\}$.
Note that $\{x,y\} \in E(G)$.
If $x, y \in M$, then $T_{i-1} = S_{i-1} - x + v = S_{i} - y + v = T_{i}$,
and thus $T_{i-1} \stackrel{G}{\sevstep} T_{i}$.
If $x, y \notin M$, then $T_{i-1} \setminus T_{i} = \{x\}$ and $T_{i} \setminus T_{i-1} = \{y\}$.
We have $T_{i-1} \stackrel{G}{\onestep} T_{i}$ as $\{x,y\} \in E(G)$.
Now by symmetry assume that $x \in M$ and $y \notin M$.
Since $\{x,y\} \in E(G)$ and $v,x \in M$, we have $\{v,y\} \in E(G)$.
Since $T_{i-1} \setminus T_{i} = \{v\}$ and $T_{i} \setminus T_{i-1} = \{y\}$,
we have $T_{i-1} \stackrel{G}{\onestep} T_{i}$.
\qed
\end{proof}

Now we are ready to present our algorithm for the $\TS$-reachability problem.
\begin{theorem}
[$\TS$]
\label{thm:TS}
There is an algorithm which, given a graph $G$ and two independent sets $S,T$,
decides if $S \sevstep T$ in time $O^*(2^{\mw(G)})$.
\end{theorem}
\begin{proof}
We first check the size of each independent set and return ``no'' if $|S| \ne |T|$.
We return ``yes'' if $|S| = |T| = 0$. Otherwise, we check the connectivity of $G$.
If $G$ is not connected, then we can solve each component independently by Observation~\ref{obs:TS-component}.
We return ``yes'' if and only if all executions on the components return ``yes.''

Now assume that $G$ is connected. We compute a modular decomposition of $G$
which gives a partition of $V$ into $r \le \mw(G)$ modules $M_{1}, \dots, M_{r}$.
(We may assume that $r \ge 2$ since otherwise $G$ has at most $\mw(G)$
vertices and the claimed running time is trivial in that case.)
We then check whether there is $i \in [r]$ such that $|S \cap M_{i}| \ge 2$.
If such an $i$ exists, then by Lemma~\ref{lem:TS-at-least-2}, 
$S \stackrel{G}{\sevstep} T$ if and only if $T \cap N(M_{i}) = \emptyset$ and $S \stackrel{G - N(M_{i})}{\sevstep} T$.
Hence, if $T \cap N(M_{i}) \ne \emptyset$, then we return ``no'';
otherwise, we recursively check whether $S \stackrel{G - N(M_{i})}{\sevstep} T$.

In the following, we assume that the instance is not caught by the tests above.
That is, $G$ is connected and $|S \cap M_{i}| \le 1$ for each $i \in [r]$.
We also assume that $|T \cap M_{i}| \le 1$ holds for each $i \in [r]$,
since otherwise the answer is ``no'' by Corollary~\ref{cor:TS-at-most-1}.

We then exhaustively apply Lemmas~\ref{lem:TS-at-most-1}, \ref{lem:TS-both1}, \ref{lem:TS-both2} to remove ``irrelevant'' vertices.
When we apply Lemma~\ref{lem:TS-both2}, we remember which module is involved
as we need it later for the generalized reachability test.
After all, we end up with a reduced instance where each module is of size 1
and the list of modules that are used for applying Lemma~\ref{lem:TS-both2}.
Let $H$ be the reduced graph with the vertex set $\{v_{1}, \dots, v_{r}\}$ where $v_{i} \in M_{i}$,
$S', T'$ the modified independent sets,
and $\mathcal{L}$ the list  of modules used in Lemma~\ref{lem:TS-both2}.
We construct the auxiliary graph $\mathcal{G}$ as follows:
we set the vertex set $V(\mathcal{G})$ to be the set of all size-$|S'|$ independent sets in $H$.
Two sets $X, Y \in V(\mathcal{G})$ are adjacent in $\mathcal{G}$ if and only if $X \stackrel{H}{\onestep} Y$.
We then compute the component $\mathcal{C}$ of $\mathcal{G}$ that contains $S'$.
We return ``yes'' if 
\begin{itemize}
  \item $\mathcal{C}$ contains $T'$, and
  \item for each module $M_{i}$ in $\mathcal{L}$, there is $Z \in \mathcal{C}$ with $Z \cap M_{i} = \emptyset$
  (equivalently, $v_{i} \notin Z$).
\end{itemize}
Otherwise, we return ``no.''
Note that $|V(\mathcal{G})| \le \binom{r}{|S'|} \le 2^{r}$
and $|E(\mathcal{G})| \le |V(\mathcal{G})| \cdot |V(H)| \le 2^{r} r$.

The correctness of the algorithm follows directly from the correctness of facts used
(that is, Observation~\ref{obs:TS-component},
  Lemmas~\ref{lem:TS-at-least-2}, \ref{lem:TS-at-most-1}, \ref{lem:TS-both1}, \ref{lem:TS-both2},
  and Corollary~\ref{cor:TS-at-most-1}).
We next consider the running time.
The algorithm first reduces the instance to a collection of instances of size $O(r)$.
This phase runs in time polynomial in the number of vertices of $G$.
Then the algorithm solves each instance in time $O^{*}(2^{r})$.
Thus the total running time is $O^{*}(2^{r})$.
\qed
\end{proof}


\bibliographystyle{abbrv}
\bibliography{mw}

\providecommand{\noopsort}[1]{}
\begin{thebibliography}{10}

\bibitem{BelmonteKLMOS19}
R.~Belmonte, E.~J. Kim, M.~Lampis, V.~Mitsou, Y.~Otachi, and F.~Sikora.
\newblock Token sliding on split graphs.
\newblock In {\em {STACS}}, volume 126 of {\em LIPIcs}, pages 13:1--13:17.
  Schloss Dagstuhl - Leibniz-Zentrum fuer Informatik, 2019.

\bibitem{BonamyB17}
M.~Bonamy and N.~Bousquet.
\newblock Token sliding on chordal graphs.
\newblock In {\em {WG} 2017}, volume 10520 of {\em LNCS}, pages 127--139, 2017.

\bibitem{Bonsma16}
P.~S. Bonsma.
\newblock Independent set reconfiguration in cographs and their
  generalizations.
\newblock {\em Journal of Graph Theory}, 83(2):164--195, 2016.

\bibitem{BonsmaKW14}
P.~S. Bonsma, M.~Kaminski, and M.~Wrochna.
\newblock Reconfiguring independent sets in claw-free graphs.
\newblock In {\em {SWAT}}, volume 8503 of {\em LNCS}, pages 86--97. Springer,
  2014.

\bibitem{BousquetMP17}
N.~Bousquet, A.~Mary, and A.~Parreau.
\newblock Token jumping in minor-closed classes.
\newblock In {\em {FCT}}, volume 10472 of {\em LNCS}, pages 136--149. Springer,
  2017.

\bibitem{CournierH94}
A.~Cournier and M.~Habib.
\newblock A new linear algorithm for modular decomposition.
\newblock In {\em {CAAP}}, volume 787 of {\em LNCS}, pages 68--84. Springer,
  1994.

\bibitem{CyganFKLMPPS15}
M.~Cygan, F.~V. Fomin, L.~Kowalik, D.~Lokshtanov, D.~Marx, M.~Pilipczuk,
  M.~Pilipczuk, and S.~Saurabh.
\newblock {\em Parameterized Algorithms}.
\newblock Springer, 2015.

\bibitem{DemaineDFHIOOUY15}
E.~D. Demaine, M.~L. Demaine, E.~Fox{-}Epstein, D.~A. Hoang, T.~Ito, H.~Ono,
  Y.~Otachi, R.~Uehara, and T.~Yamada.
\newblock Linear-time algorithm for sliding tokens on trees.
\newblock {\em Theor. Comput. Sci.}, 600:132--142, 2015.

\bibitem{FominLMT18}
F.~V. Fomin, M.~Liedloff, P.~Montealegre, and I.~Todinca.
\newblock Algorithms parameterized by vertex cover and modular width, through
  potential maximal cliques.
\newblock {\em Algorithmica}, 80(4):1146--1169, 2018.

\bibitem{Fox-EpsteinHOU15}
E.~Fox{-}Epstein, D.~A. Hoang, Y.~Otachi, and R.~Uehara.
\newblock Sliding token on bipartite permutation graphs.
\newblock In {\em {ISAAC}}, volume 9472 of {\em LNCS}, pages 237--247.
  Springer, 2015.

\bibitem{GajarskyLO13}
J.~Gajarsk{\'{y}}, M.~Lampis, and S.~Ordyniak.
\newblock Parameterized algorithms for modular-width.
\newblock In {\em {IPEC}}, volume 8246 of {\em LNCS}, pages 163--176. Springer,
  2013.

\bibitem{HabibP10}
M.~Habib and C.~Paul.
\newblock A survey of the algorithmic aspects of modular decomposition.
\newblock {\em Computer Science Review}, 4(1):41--59, 2010.

\bibitem{HearnD05}
R.~A. Hearn and E.~D. Demaine.
\newblock {PSPACE}-completeness of sliding-block puzzles and other problems
  through the nondeterministic constraint logic model of computation.
\newblock {\em Theor. Comput. Sci.}, 343(1-2):72--96, 2005.

\bibitem{vandenHeuvel13}
J.~{\noopsort{Heuvel}van den Heuvel}.
\newblock The complexity of change.
\newblock In S.~R. Blackburn, S.~Gerke, and M.~Wildon, editors, {\em Surveys in
  Combinatorics 2013}, volume 409 of {\em London Mathematical Society Lecture
  Note Series}, pages 127--160. Cambridge University Press, 2013.

\bibitem{HoangU16}
D.~A. Hoang and R.~Uehara.
\newblock Sliding tokens on a cactus.
\newblock In {\em {ISAAC}}, volume~64 of {\em LIPIcs}, pages 37:1--37:26.
  Schloss Dagstuhl - Leibniz-Zentrum fuer Informatik, 2016.

\bibitem{ItoDHPSUU11}
T.~Ito, E.~D. Demaine, N.~J.~A. Harvey, C.~H. Papadimitriou, M.~Sideri,
  R.~Uehara, and Y.~Uno.
\newblock On the complexity of reconfiguration problems.
\newblock {\em Theor. Comput. Sci.}, 412(12--14):1054--1065, 2011.

\bibitem{ItoKOSUY14}
T.~Ito, M.~Kaminski, H.~Ono, A.~Suzuki, R.~Uehara, and K.~Yamanaka.
\newblock On the parameterized complexity for token jumping on graphs.
\newblock In {\em {TAMC}}, volume 8402 of {\em LNCS}, pages 341--351. Springer,
  2014.

\bibitem{ItoKO14}
T.~Ito, M.~J. Kaminski, and H.~Ono.
\newblock Fixed-parameter tractability of token jumping on planar graphs.
\newblock In {\em {ISAAC}}, volume 8889 of {\em LNCS}, pages 208--219.
  Springer, 2014.

\bibitem{KaminskiMM12}
M.~Kami\'{n}ski, P.~Medvedev, and M.~Milani\v{c}.
\newblock Complexity of independent set reconfigurability problems.
\newblock {\em Theor. Comput. Sci.}, 439:9--15, 2012.

\bibitem{Lampis12}
M.~Lampis.
\newblock Algorithmic meta-theorems for restrictions of treewidth.
\newblock {\em Algorithmica}, 64(1):19--37, 2012.

\bibitem{LokshtanovM18}
D.~Lokshtanov and A.~E. Mouawad.
\newblock The complexity of independent set reconfiguration on bipartite
  graphs.
\newblock In {\em SODA 2018}, pages 185--195, 2018.

\bibitem{MouawadN0SS17}
A.~E. Mouawad, N.~Nishimura, V.~Raman, N.~Simjour, and A.~Suzuki.
\newblock On the parameterized complexity of reconfiguration problems.
\newblock {\em Algorithmica}, 78(1):274--297, 2017.

\bibitem{Nishimura18}
N.~Nishimura.
\newblock Introduction to reconfiguration.
\newblock {\em Algorithms}, 11(4):52, 2018.

\bibitem{TedderCHP08}
M.~Tedder, D.~G. Corneil, M.~Habib, and C.~Paul.
\newblock Simpler linear-time modular decomposition via recursive factorizing
  permutations.
\newblock In {\em {ICALP} {(1)}}, volume 5125 of {\em LNCS}, pages 634--645.
  Springer, 2008.

\bibitem{Wrochna18}
M.~Wrochna.
\newblock Reconfiguration in bounded bandwidth and tree-depth.
\newblock {\em J. Comput. Syst. Sci.}, 93:1--10, 2018.

\end{thebibliography}

\end{document}